\documentclass[a4paper, 11pt]{article}

\usepackage[toc,page]{appendix}

\usepackage{hyperref}
\usepackage{thm-restate}

\usepackage{array} 

\usepackage{setspace}
\usepackage[margin=1in]{geometry}

\usepackage{authblk}

\usepackage{titling}

\usepackage{amsfonts}
\usepackage{amsmath}

\usepackage{graphicx} 
\graphicspath{ {images/} }
\usepackage{subcaption}

\usepackage[thinlines]{easytable}

\usepackage{amsthm}
\newtheorem{definition}{Definition}[section]
\newtheorem{theorem}{Theorem}[section]
\newtheorem{lemma}{Lemma}[section]
\newtheorem{corollary}{Corollary}[section]
\newtheorem{remark}{Remark}[section]

\newenvironment{rcases}
  {\left.\begin{aligned}}
  {\end{aligned}\right\rbrace}

\newcommand{\ignore}[1]{}

\let\svthefootnote\thefootnote
\newcommand\freefootnote[1]{%
  \let\thefootnote\relax%
  \footnotetext{#1}%
  \let\thefootnote\svthefootnote%
}

\newcommand{\R}{\mathcal{R}}
\newcommand{\A}{\mathcal{A}}

\newcolumntype{C}{>{\centering\arraybackslash}p{2.5cm}}

\usepackage{subfiles}

\begin{document}

\title{2D Fractional Cascading on Axis-aligned Planar Subdivisions}
\author{Peyman Afshani$^*$}
\author{Pingan Cheng$^*$}
\freefootnote{$^*$Supported by DFF (Det Frie Forskningsr\" ad) of Danish Council for Independent Research under grant ID DFF$-$7014$-$00404.}
\affil{Department of Computer Science, Aarhus University, Denmark}
\affil{\{peyman, pingancheng\}@cs.au.dk}
\renewcommand\Affilfont{\itshape\small}
\maketitle

\medskip

\begin{abstract}
Fractional cascading is one of the influential and important techniques in data structures,
as it provides a general framework for solving a
common important problem: the iterative search problem.
In the problem, the input is a 
graph $G$ with constant degree.
Also as input, we are given a set of values for every vertex of $G$. 
The goal is to preprocess $G$ such that when we are given a query value $q$, and a connected
subgraph $\pi$ of $G$, we can find the predecessor of $q$ in all the sets associated with the
vertices of $\pi$. 
The fundamental result of fractional cascading, by Chazelle and Guibas, is that there exists a
data structure that uses linear space and it can answer queries in
$O(\log n + |\pi|)$ time, at essentially 
constant time per predecessor~\cite{chazelle1986fractionali}.
While this technique has received plenty of attention in the past decades, 
an almost quadratic space lower bound for ``two-dimensional fractional
cascading'' by Chazelle and Liu in STOC 2001~\cite{chazelle2004lower}
has convinced the researchers that fractional cascading is fundamentally a one-dimensional 
technique.

In two-dimensional fractional cascading, the input includes a planar subdivision 
for every vertex of $G$ and the query is a point $q$ and a subgraph $\pi$ and the goal is to 
locate the cell containing $q$ in all the subdivisions associated with the vertices of $\pi$. 
In this paper, we show that it is actually possible to circumvent the lower bound of 
Chazelle and Liu for axis-aligned planar subdivisions.
\ignore{ 
vertex attaching an axis-aligned planar subdivision. We assume the total size
of all subdivisions to be $n$ and we want to build a data structure such that
given any query point $q$ and a connected subgraph $\pi$, all regions of
subdivisions attached to nodes in $\pi$ containing $q$ can be reported
efficiently. Despite its own interest, this problem appears in geographic
information systems as well as many classic computational geometry problems in
high dimensions, e.g., 4D dominance range reporting, 3D point locations on
orthogonal subdivisions, 3D vertical ray-shooting on axis-parallel polyhedra.
}
We present a number of upper and lower bounds which reveal that in two-dimensions, the problem has a much richer structure.
When $G$ is a tree and $\pi$ is a path, then queries can be answered in 
$O(\log{n}+|\pi|+\min\{|\pi|\sqrt{\log{n}},\alpha(n)\sqrt{|\pi|}\log{n}\})$ time using linear space
where $\alpha$ is an inverse  Ackermann function; surprisingly, we show both branches of this
bound are tight, up to the inverse Ackermann factor. 
When $G$ is a general graph or when $\pi$ is a general subgraph, then the query bound
becomes $O(\log n + |\pi|\sqrt{\log n})$ and this bound is once again tight in both cases. 
\end{abstract}

\section{Introduction}
Fractional cascading \cite{chazelle1986fractionali} is one of the widely
used tools in data structures as it provides a general framework for solving a
common important problem: the iterative search problem, i.e., the problem of
finding the predecessor of a single value $q$ in multiple data sets. 
In the problem, we are to preprocess a degree-bounded ``catalog''
graph $G$ where each vertex represents an input set of values from a totally ordered universe $U$;
the input sets of different vertices of $G$ are completely unrelated. 
Then, at the query time, given a value $q \in U$ and a connected subgraph $\pi$ of $G$,
the goal is to find the predecessor of $q$ in the \ignore{lists} sets that correspond to the vertices of $\pi$. 
The fundamental theorem of fractional cascading is that one can build a data structure of linear
size such that the queries can be answered in $O(\log n + |\pi|)$
time\footnote{All logs are base 2 unless otherwise specified.}, essentially giving us
constant search time per predecessor after investing an initial $O(\log n)$ search time~\cite{chazelle1986fractionali}.
Many problems benefit from this technique~\cite{chazelle1986fractionalii}  since they need to
solve the iterative search problem as a base problem. 

Given its importance, it is not surprising that many have
attempted to  generalize this technique:
The first obvious direction is to consider the dynamic version of the problem by allowing
insertions or deletions into the \ignore{lists} sets of the vertices of $G$.
In fact, Chazelle and Guibas themselves consider this~\cite{chazelle1986fractionali} and they show that with 
$O(\log n)$  amortized time per update, one can obtain $O(\log n + |\pi|\log\log n)$ query time.
Later, Mehlhorn and N\" aher improve the update time to $O(\log\log n)$ amortized time~\cite{mn90} and
then Dietz and Raman~\cite{dr91} remove the amortization. 
 There is also some attention given to optimize the dependency of the query time
 on the  maximum degree of graph $G$~\cite{gk09}.
     
The next obvious generalization is to consider the higher dimensional versions of the problem.
Here, each vertex of $G$ is associated with an input subdivision and the goal is to locate a given
query point $q$ on every subdivision associated with the vertices of $\pi$.
Unfortunately, here we run into an immediate roadblock already in two dimensions:
After listing a number of potential applications of two-dimensional fractional cascading, 
Chazelle and Liu \cite{chazelle2004lower} ``dash all such hopes'' by showing an 
$\tilde{\Omega}(n^2)$ space lower bound\footnote{The $\tilde{\Omega}$ notation hides polylogarithmic factors.} 
in the pointer-machine model for any data structure that can answer queries
in $O(\log^{O(1)}n + |\pi|)$ time. 
Note that this lower bound can be generalized to also give a 
$\Omega(n^{2-\varepsilon})$ space lower bound for data structures with 
$O(\log^{O(1)}n) + o( |\pi|\log n)$ query time. 
As far as we can tell, progress in this direction was halted due to this negative result since
the trivial solution already gives the $O(|\pi| \log n)$ query time, by just building
individual point location data structures for each subdivision.

We observe that the lower bound of Chazelle and Liu does not apply to orthogonal subdivisions,
a very important special case of planar point location problem. 
Many geometric problems need to solve this base problem, e.g., 
4D orthogonal dominance range reporting~\cite{afshani2012higher, afshani2014deterministic}, 
3D point location in orthogonal subdivisions~\cite{rahul2014improved}, 
some 3D vertical ray-shooting problems~\cite{de1992two}. 
In geographic information systems, it is very common to overlay planar subdivisions 
describing different features of a region to generate a complete map. 
Performing point location queries on such maps 
corresponds to iterative point locations on a series of subdivisions.

Motivated by this observation, 
we systematically study the generalization of fractional cascading to two dimensions, when restricted to
orthogonal subdivisions. We obtain a number of interesting results, including both upper and lower bounds
which show most of our results are tight except for the general path queries of trees 
where the bound is tight up to a tiny inverse Ackermann factor~\cite{cormen2009introduction}. 

\paragraph{The problem definition}
The formal definition of the problem is as follows.
The input is a degree-bounded connected graph $G=(V, E)$
where each vertex $v\in{V}$ is associated with an axis-aligned planar subdivision. 
Let $n$ be the total number of vertices, edges, and faces in the subdivisions, which
we call the \textit{graph subdivision complexity}. 
We would like to build a data structure such that given
a query $(q, \pi)$, where $q$ is a query point and $\pi$ is a connected
subgraph of $G$,  we can locate $q$ in all
the subdivisions induced by vertices of $\pi$ efficiently. We call this problem
2D Orthogonal Fractional Cascading (2D OFC). 

\subsection{Related Work}
While the negative result of Chazelle and Liu \cite{chazelle2004lower} stops any progress on the general problem
of two-dimensional fractional cascading, there have been other results that can be seen as special cases of 
two-dimensional fractional cascading. 
For example, Chazelle et al. \cite{chazelle1994ray} improved the result of ray
shooting in a simple polygon by a $\log{n}$ factor. In a ``geodesically triangulated''
subdivision of $n$ vertices, they showed it is possible to locate all the
triangles crossed by a ray in $O(\log{n})$ time instead of $O(\log^2{n})$,
which resembles 2D fractional cascading. However, their solution relies heavily
on the characteristic of geodesic triangulation and cannot be generalized to
other problems.
Chazelle's data structure for the rectangle stabbing problem~\cite{c86} can also 
be viewed as a restricted form of two-dimensional fractional cascading where 
$\pi=G$.

In recent years, interestingly, a technique similar to 2D fractional cascading
has been used to improve many classical computational geometry data structures.
While working on the 4D dominance range reporting problem, Afshani et al.~\cite{afshani2012higher}
are implicitly performing iterative point location queries along a path of a balanced
binary tree on somewhat specialized subdivsions in $O(\log^{3/2}n)$ total time.
Later Afshani et al.~\cite{afshani2014deterministic} studied an offline variant of this problem,
and they presented a linear sized data structure that achieves optimal query time.
The same idea is used to improve the result of 3D point location in orthogonal subdivisions. 
In that probelm, Rahul \cite{rahul2014improved} obtained another data structure with 
$O(\log^{3/2}n)$ query time.
\ignore{
Later Afshani et al.~\cite{afshani2014deterministic} studied an offline version of tree point
location problem and gave an optimal $O(n+k\log{n})$ query time and $O(n)$
space data structure, where $k$ is the number of query points and $n$ is the
subdivision size of a tree of height $O(\log{n})$. 
}

Another related problem is  the ``unrestricted'' version of fractional cascading where essentially
$\pi$ can be an arbitrary subgraph of $G$, instead of a connected subgraph. 
In one variant, we are given a set $L$ of categories and a set $S$ of $n$ points in $d$ dimensional space
where each point belongs to one of the categories. 
The query is given by a $d$-dimensional rectangle $r$ and a subset $Q \subset L$ of the categories. 
We are asked
to report the points in $S$ contained in $r$ and belonging to the categories in $Q$. 
In 1D, Chazelle and Guibas \cite{chazelle1986fractionalii} provided a
$O(|Q|\log{\frac{|L|}{|Q|}}+\log{n}+k)$ query time and linear size data
structure, where $k$ is the output size, together with a restricted lower bound. 
Afshani et al.~\cite{afshani2014concurrent} strenghtened the lower bound and presented several
data structures for three-sided queries in two-dimensions. 
Their data structures match the lower bound within an inverse Ackermann factor for the general case.

\subsection{Our Results}
\label{sec:ourresults}
We study 2D OFC in a pointer machine model of computation.
Some of our bounds involve inverse Ackermann functions.
The particular definition that we use is the following.
We define $\alpha_2(n) = \log n$ and then we define
$\alpha_i(n)= \alpha_{i-1}^*(n)$, meaning, it's the number of times
we need to apply the $\alpha_{i-1}(\cdot)$ function to $n$ until we reach a fixed constant. 
$\alpha(n)$ corresponds to the value of $i$ such that 
$\alpha_i(n)$ is at most a fixed constant. 
Our results are summarized in \autoref{tab:results}. 
\ignore{
  $$
  f(i,n)=
  \begin{cases}
  f(i,f(i-1,n))+1 & i>0,\\
  f(0,n)=\log{n} & n>1,\\
  f(i,1)=1.\\
  \end{cases}
  $$

  and

  $$
  \tau(k,n)=
  \begin{cases}
  \min\{i:\log^2n/f^2(i,n)>k\} & \Omega(\log{n}) \leq k \leq o(\log^2{n}),\\
  1 & o.w.\\
  \end{cases}
  $$
} 

{\small
\begin{table}[h!]
\centering
\setlength\extrarowheight{2.5pt}
\caption{Our Results}
\label{tab:results}
\begin{tabular}{ | c | m{2.8cm}| m{1.8cm} | m{6.5cm} | m{2cm} |} 
\hline
\bf{Graph}& \bf{Query} & \bf{Space} & \bf{Query Time} & \bf{Tight?} \\
\hline
Tree  & Path & $O(n\alpha_c(n))$ & $O(\min\{|\pi|\sqrt{\log{n}},c\sqrt{|\pi|}\log{n}\}+\log{n}+|\pi|)$ & Up to $\alpha_c(n)$ factor\\
\hline
Tree  & Path & $O(n)$ & $O(\min\{|\pi|\sqrt{\log{n}},\alpha(n)\sqrt{|\pi|}\log{n}\}$ $+ \log{n} + |\pi|)$ & Up to $\alpha(n)$ factor \\
\hline
Tree  & Subtree & $O(n)$ & $O(\log{n}+|\pi|\sqrt{\log{n}})$ & yes \\
\hline
Graph & Path / Subgraph & $O(n)$ & $O(\log{n}+|\pi|\sqrt{\log{n}})$ & yes \\
\hline
\end{tabular}
\end{table}
}

Our results show some very interesting behavior. 
First, by looking at the last two rows of \autoref{tab:results},
we can see that we can always do better than the na\" ive solution
by a $\sqrt{\log n}$ factor. Furthermore, this is tight. We show matching
query lower bounds both when $G$ can be an arbitrary graph but with $\pi$ being restricted to a 
path and also when $G$ is a tree but $\pi$ is allowed to be any subtree of $G$.
Second, 
when $G$ is a tree and $\pi$ is a path we get some variation depending on the length of the
query path. 
When $\pi$ is of length at most $\frac{\log n}{2}$, then we can answer queries in 
$O(|\pi|\sqrt{\log n})$ time, but when $\pi$ is longer than $\frac{\log n}{2}$, we obtain
the query bound of $O(\sqrt{|\pi|}\log n)$ (ignoring some inverse Ackermann factors).
Furthermore, we give two lower bounds that show both of these branches are tight!
When $\pi$ is very long, longer than $\frac{\log^2 n}{2}$, then the query bound becomes
$O(|\pi|)$ which is also clearly optimal.

\section{Preliminaries}

In this section, we introduce some geometric preliminaries and 
present the tools we will use to build the data structures and to prove the lower bounds.

\subsection{Geometric Preliminaries}
First we review the definition of planar subdivisions.

\begin{definition}
A graph is said to be a planar graph if it can be embedded in the plane without crossings.
A planar subdivision is a planar embedding of a planar graph where all the edges are straight line segments.
The complexity of a planar subdivision is the sum of the number of vertices, edges, and faces of the subdivision.
\end{definition}

Planar point location, defined below, is one classical problem related to planar subdivisions:

\begin{definition}
Given a planar subdivision $S$ of complexity $n$,
in the planar point location problem,
we are asked to preprocess $S$ such that given any query point $q$ in the plane, 
we can find the face $f$ in $S$ containing $q$ efficiently.
\end{definition}

Note that we can assume that the subdivision is enclosed by a bounding box.
There are several different ways to solve the planar point location problem optimally
in $O(\log n)$ query time and $O(n)$ space, see~\cite{toth2017handbook} for details.
One simple solution uses trapezoidal decomposition, see~\cite{bcko08} for a detailed introduction.
Roughly speaking, given a planar subdivision $S$ enclosed by a bounding box $R$,
we construct a trapezoidal decomposition of the subdivision by extending two rays
from every vertex of $S$, one upwards and one downwards.
The rays stop when they hit an edge in $S$ or the boundary of $R$.
The faces of the subdivision we obtain after this transform will be only trapezoids.
\autoref{fig:trape-example} gives an example of trapezoidal decomposition.
A crucial property of trapezoidal decomposition is that it increases the
complexity of the subdivision by only a constant factor.

\begin{figure}[h]
     \centering
     \begin{subfigure}[b]{0.45\textwidth}
         \centering
         \includegraphics[width=0.7\textwidth]{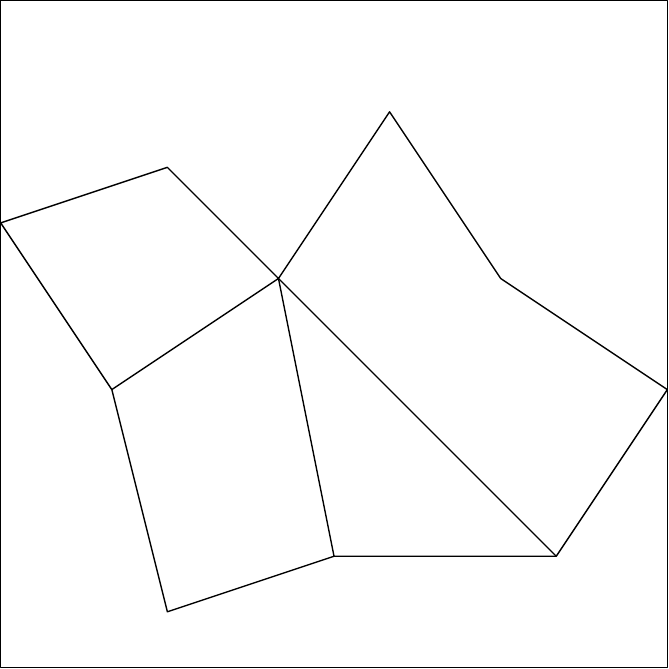}
         \caption{A Planar Subdivision}
         \label{fig:planar-subd}
     \end{subfigure}
     \hfill
     \begin{subfigure}[b]{0.45\textwidth}
         \centering
         \includegraphics[width=0.7\textwidth]{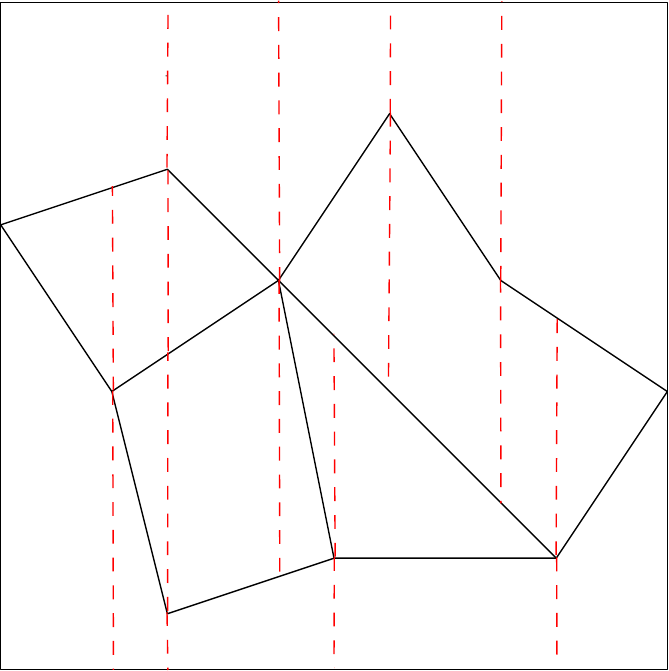}
         \caption{After Trapezoidal Decomposition}
         \label{fig:trape-decomp}
     \end{subfigure}
     \caption{Example of Trapezoidal Decomposition}
     \label{fig:trape-example}
\end{figure}

We also review some concepts related to cuttings.

\begin{definition}
Given a set $H$ of $n$ hyperplanes in the plane, a $(1/r)$-cutting, $1\le r\le n$,
is a set of (possibly open) disjoint simplices that together cover the entire plane
such that each simplex intersects $O(n/r)$ hyperplanes of $H$.
For each simplex in the cutting, the set of all hyperplanes of $H$ intersecting it is called
the conflict list of that simplex.
\end{definition}

$(1/r)$-cuttings are important in computational geometry as 
they enable us to apply the divide-and-conquer paradigm in higher dimensions.
The following theorem by Chazelle~\cite{chazelle1993cutting},
after a series of work in the computational geometry community
~\cite{matouvsek1991cutting,matousek1995approximations,agarwal1990partitioning,agarwal1991geometric,chazelle1990deterministic},
shows the existence of $(1/r)$-cuttings of small size
and an efficient deterministic algorithm computing $(1/r)$-cuttings.

\begin{theorem}[Chazelle~\cite{chazelle1993cutting}]
Given a set $H$ of $n$ hyperplanes in the plane,
there exists a $(1/r)$-cutting, $1\le r\le n$, of size $O(r^2)$, which is optimal.
We can find the cutting and the corresponding conflict lists in $O(nr)$ time.
\end{theorem}

In this paper, we will use intersection sensitive $(1/r)$-cuttings
which is a generalization of $(1/r)$-cuttings.
The following theorem is given by de Berg and Schwarzkopf~\cite{berg1995cuttings}.

\begin{theorem}[de Berg and Schwarzkopf~\cite{berg1995cuttings}]
\label{thm:intersection-sensitive-cutting}
Given a set $H$ of $n$ line segments in the plane with $A$ intersections,
we can construct a $(1/r)$-cutting, $1\le r\le n$, of size $O(r+Ar^2/n^2)$.
 We can find the cutting and the corresponding conflict lists in time $O(n\log r + Ar/n)$
 using a randomized algorithm.
\end{theorem}

Note that by the construction of generalized cuttings, see~\cite{berg1995cuttings} for detail,
the following corollary follows directly from \autoref{thm:intersection-sensitive-cutting},

\begin{corollary}
\label{cor:orth-sen-cutting}
Given an axis-aligned planar subdivision of complexity $n$,
we can construct a $(1/r)$-cutting, $1\le r\le n$, of size $O(r)$.
More specifically, each cell of the cutting is an axis-aligned rectangle
and the size of the conflict list of every cell is bounded by $O(n/r)$.
We can find the cutting and the corresponding conflict lists in time $O(n\log n)$
using a randomized algorithm.
\end{corollary}

\subsection{Rectangle Stabbing}

In $d$-dimensional rectangle stabbing problem, we are given a set of $n$
$d$-dimensional axis-parallel rectangles, our task is to build a data structure
such that given a query point $q$, we can report the rectangles containing the
query point efficiently. 
As noted earlier, Chazelle~\cite{c86} provides an optimal solution in two-dimensions, a linear-sized
data structure that can answer queries in $O(\log n + t)$ time where $t$ is the output size.
The following lemma by Afshani et al.~\cite{afshani2012higher} establishes an upper bound of this problem
and it is obtained by a basic application of range trees~\cite{Bentley.79} with large fan-out and Chazelle's data structure. 

\begin{lemma}[Afshani et al. \cite{afshani2012higher}]
\label{lem:stabbing}
We can answer $d$ dimensional rectangle stabbing queries in time $O(\log{n}\cdot(\log{n}/\log{H})^{d-2}+t)$
using space $O(nH\log^{d-2}n)$, where $n$ is the number of rectangles, $t$ is the output size, and $H\geq{2}$ is any parameter.
\end{lemma}

\subsection{A Pointer Machine Lower Bound Framework}
We will use the pointer machine lower bound framework of Afshani~\cite{afshani2012improved}.
The framework deals with an abstract ``geometric stabbing problem'' which is defined by a set $\R$ of
``ranges'' and a set $U$ of queries. 
An instance of the geometric stabbing problem is given by a set $R \subset \R$ of $n$ ``ranges''
and the goal is to preprocess $R$ to answer queries $q$. 
Given $R$, an element $q\in U$ (implicitly) defines a subset $R_q\subset R$ and the 
data structure is to output the elements of $R_q$. 
However, the data structure is restricted to operate in the (strengthened) pointer machine model of computation where
the memory is a directed graph $M$ consisting  of ``cells'' 
where each cell can store an element of $R$ as well as two pointers to other memory cells.
At the query time, the algorithm must find a connected subgraph $M_q$ of $M$ where each element of $R_q$
is stored in at least one memory cell of $M_q$.
The size of $M$ is a lower bound on the space complexity of the data structure and the size of $R_q$ is a lower bound
on the query time. 
However, the lower bound model allows for unlimited computation and allows the data structure to have complete information
about the problem instance; the only bottleneck is being able to navigate to the cells storing the output elements.
In addition, the framework assumes that we have a measure $\mu$
such that $\mu(U) = 1$. 
We need a slightly more precise version of the lower bound framework where the dependency on a certain ``constant'' is
made explicit. 

\begin{restatable}{theorem}{thmframework}\label{thm:framework}
  Assume, we have an algorithm that given any input instance $R \subset \R$ of $n$ ranges, it can store $R$ in 
  a data structure of size $S(n)$ such that given any query $q \in U$,  it can answer the query 
  in $Q(n) + \gamma |R_q|$ time.

  Then, suppose we can construct an input set
  $R \subset \R$ of $n$ ranges such that the following two conditions are satisfied: 
  (i) every query point $q\in U$ is contained in exactly $|R_q|=t$ ranges
  and $\gamma t \ge Q(n)$; (ii) there exists a value $v$ such that
  for any two ranges $r_1, r_2 \in R$, $\mu(\left\{ q \in U | r_1, r_2 \in R_q \right\})$ is
  well-defined and is upper bounded by $v$.
  Then, we must have $S(n)=\Omega(tv^{-1}/2^{O(\gamma)})=\Omega(Q(n)v^{-1}/2^{O(\gamma)})$. 
  
\end{restatable}

For the proof of this theorem, we refer the readers to \autoref{sec:frameworkproof}.
In our applications, $\mu$ will basically be the Lebesgue measure and $U$ will be the unit cube. 

\section{Queries on Catalog Paths}

In this section, we give a simple solution for when the catalog graph is a path.
It will be used as a building block for later data structures.

\begin{theorem}
\label{thm:path}
Consider a catalog path $G$, in which each vertex is associated with a planar subdivision.
Let $n$ be the total complexity of the subdivisions. 
We can construct
a data structure using $O(n)$ space such that given any query $(q, \pi)$, where
$q$ is a query point and $\pi$ is a subpath, all regions containing $q$ along
$\pi$ can be reported in time $O(\log{n}+|\pi|)$.
\end{theorem}
\begin{proof}
  We can convert each subdivision into a set of disjoint rectangles of total size 
  $O(n)$ using
  trapezoidal decomposition~\cite{bcko08}.
  Then, we partition $G$ into $m=\lceil |G|/\log n \rceil$ paths, $G_1, \cdots, G_m$ 
  where each path except potentially for $G_m$ has size $\log n$ and $G_m$ has size at most $\log n$. 

  Now we use an observation that was also made in previous papers~\cite{aal10,afshani2012higher,rahul2014improved}:
  when $H=G$, the two-dimensional fractional cascading can be reduced to rectangle stabbing. 
  As a result, for each $G_i$, $1 \le i \le m$, we collect all the rectangles of its 
  subdivisions and build a 2D rectangle stabbing data structure on them. 
  By~\autoref{lem:stabbing} this requires $O(n)$ space. 
  Now given a query subpath of length $|\pi|$, we use the rectangle stabbing data structures on the subdivisions of
  each $G_i$ as long as $|G_i \cap \pi| > 0$.
  Since $\pi$ is a path, for at most two indices $i$ we will have $0 < |G_i \cap \pi| < \log n$ and for the rest
  $|G_i \cap \pi| = |G_i| = \log n$.
  This gives us $O(\log n + |\pi|)$ query time. 
\end{proof}

\section{Path Queries on Catalog Trees}

Now we consider answering path queries on catalog trees. 
We first show optimal data structures for trees of different heights. 
It turns out we need different data structures to achieve optimality when heights differ. 
We then present a data structure using $O(n\alpha_c(n))$ space that can answer path queries in 
$O(\log{n}+|\pi|+\min\{|\pi|\sqrt{\log{n}},\sqrt{|\pi|}\log{n}\})$ time and a data structure using $O(n)$ space 
answering path queries in $O(\log{n}+|\pi|+\min\{|\pi|\sqrt{\log{n}},\alpha(n)\sqrt{|\pi|}\log{n}\})$ time, 
where $c\ge 3$ is any constant and $\alpha_k(n)$ is the $k$-th function in 
the inverse Ackermann hierarchy~\cite{cormen2009introduction}  
and $\alpha(n)$ is the inverse Ackermann function~\cite{cormen2009introduction} . 
We also present lower bounds for our data structures. 
Without loss of generality, we assume the tree is a binary tree.

\subsection{Trees of height $h\le\frac{\log{n}}{2}$}
\subsubsection{The Upper Bound}
For trees of this height, we present the following upper bound.
The main idea is to use the sampling idea that is employed previously~\cite{chan:revisit,afshani2012higher}, however,
there are some main differences.
Instead of random samples or shallow cuttings, we use intersection sensitive cuttings~\cite{berg1995cuttings} and
more notably, the fractional cascading on an arbitrary tree cannot be reduced to a geometric problem such as
3D rectangle stabbing, so instead we do something else.

\begin{lemma}
\label{lem:shorttreepath}
  Consider a catalog tree of height $h\le\frac{\log{n}}{2}$ in which each vertex is associated with a planar subdivision.
  Let $n$ be the total complexity of the subdivisions. 
  We can build a data structure using $O(n)$ space such that given any query
  $(q, \pi)$, where $q$ is a query point and $\pi$ is a path, all regions
  containing $q$ along $\pi$ can be reported in time
  $O(\log{n}+|\pi|\sqrt{\log{n}})$.  \end{lemma}
\begin{proof}
  Let $r$ be a parameter to be determined later.
  Consider a planar subdivision $A_i$ and let $n_i$ be the number of rectangles in $A_i$. 
  We create an intersection sensitive $(r^2/n_i)$-cutting $C_i$ on $A_i$.
  By \autoref{cor:orth-sen-cutting}, $C_i$ contains $O(n_i/r^2)$ cells and
  each cell of $C_i$ is an axis-aligned rectangle.
  Furthermore, the conflict list size of each rectangle is $O(r^2)$.
  For each cell in $C_i$, we build an optimal point location data structure on its conflict list.
  The total space usage is linear, since total size of the conflict lists is linear. 

  Then, we consider every path of length at most $\log r$  in the catalog graph, and we call them subpaths.
  For every subpath, we collect all the cells
  of the cuttings belonging to the vertices of the subpath and 
  build a 2D rectangle stabbing data structure  on them. Since the
  degree of any vertex is bounded by $3$, each vertex is contained in at most 
  \begin{align*}
  \sum_{j=0}^{\log{r}}\sum_{i=0}^{j}3^i\cdot{3^{j-i}}=\Theta(r^{\log{3}}\log{r})
  \end{align*}
  many subpaths.
  Then the total space usage of the 2D rectangle stabbing data structures is
  bounded by $O(n\log{}r/r^{2-\log{3}})=O(n)$. 
  Given any query path $\pi$, it can be covered by $|\pi|/\log r$ subpaths. 
  For each subpath, we can find all the cells of the cuttings containing the query point in 
  $O(\log n)$ time and then perform an additional point location query on its conflict list,
  for a total of $O(\log n + (\log r)^2)$ query time per subpath.
  Thus, the query time of this data structure is bounded by
  \begin{align*}
  \frac{|\pi|}{\log{r}}(\log{n}+\log{r}\cdot{\log{r}}).
  \end{align*}
  We pick $r=2^{\sqrt{\log{n}}}$, then we obtain the desired $O(\log{n}+|\pi|\sqrt{\log{n}})$ query time.
\end{proof}

\subsubsection{The Lower Bound}
We now present a matching lower bound. We show the following:

\begin{lemma}
\label{lem:shortpathtreelower}
Assume, given any catalog tree of height $\sqrt{\log n}\le h\le \frac{\log n}{2}$ in which each vertex is associated with a planar subdivision
  with $n$ being the total complexity of the subdivisions, we can build a data structure that satisfies the following:
  it uses at most $n 2^{\varepsilon \sqrt{\log n}}$ space, for a small enough constant $\varepsilon$, and 
  it can answer 2D OFC queries $(q,\pi)$. 
  Then, its query time must be $\Omega(|\pi|\sqrt{\log{n}})$.
\end{lemma}

\begin{proof}
We will use the following idea: We consider a special 3D rectangle stabbing
problem and show a lower bound using \autoref{thm:framework}. 
We will use the 3D Lebesgue measure, denoted by $V(\cdot)$.
Then we show a reduction from this problem to a 2D OFC problem on trees to obtain the
desired lower bound.

We consider the following instance of 3D rectangle stabbing problem. 
The input $n$ rectangles are partitioned into $h$ sets of size $n/h$ each. 
The rectangles in each set are pairwise disjoint and they together tile the unit cube in 3D. 
The depth (i.e., the length of the side parallel to the $z$-axis) of rectangles in set $i=0,1,\cdots,h-1$ is $1/2^i$.
In fact, the rectangles in set $i$ can be partitioned into $2^i$ subsets
where the projection of the rectangles in the $j$-th subset, $j=0,1,\cdots,2^i-1$, onto the $z$-axis
is the interval $[\frac{j}{2^i},\frac{j+1}{2^i}]$. See \autoref{fig:rectset2} for an example.
\ignore{
\autoref{fig:rectset2} is an example of rectangles in set $2$.
Note that the depth of each rectangle is $\frac{1}{2^2}$, while the other two dimensions
can be arbitrary as long as they together tile the unit cube and the number of them is $n/h$.
Note that the rectangles in this set can be viewed as four subsets of rectangles
by cutting the $z$-dimension of the unit cube into four even intervals.
The projection of each subset into the $xy$-plane gives us an axis-aligned planar subdivision.
}

\begin{figure}[h]
  \centering
  \includegraphics[scale=0.5]{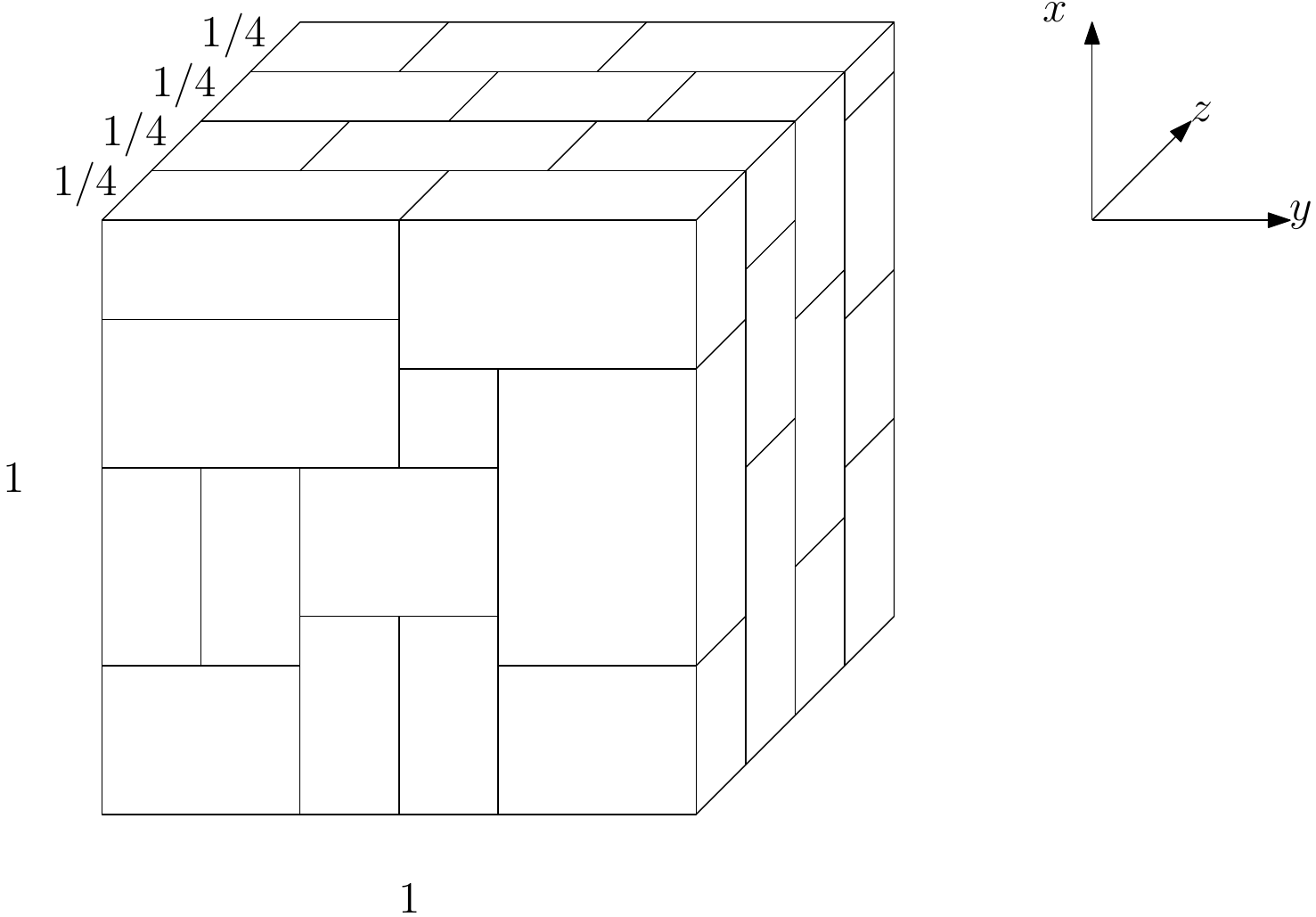}
  \caption{An example of rectangles in set 2}
  \label{fig:rectset2}
\end{figure}

We first show the reduction: 
assume, we are given an instance of the special 3D rectangle stabbing problem described above. 
We build a balanced binary tree of height $h$ on the $z$-axis as the catalog graph. Note that the
number of vertices at layer $i$ of the tree is the same as the number of subsets in set $i$. 
We project the rectangles in each subset to the $xy$-plane and obtain a 2D axis-aligned planar
subdivision. We attach each of the subdivisions to the corresponding vertices.
Consider a 2D OFC query, $(q,\pi)$ in which $\pi$ is a path that connects the root to a leaf.
We lift $q$ to a point $q'$ in 3D appropriately:
W.l.o.g., assume the leaf is the $j$-th leaf.
To obtain $q'$, we assign the $z$ coordinate $\frac{j+0.5}{2^{h-1}}$ to $q$.
By our construct, the $z$-axis projection of any rectangle in the nodes from the root to the leaf
contains the $z$ coordinate of $q'$.
This construction ensures that finding the rectangles that contain $q'$
is equivalent to performing 2D OFC query $(q,\pi)$.

\ignore{
By construction, the regions reported by a point location query from the root to a leaf
correspond to all the rectangles containing the query point after lifting it to
3D appropriately.
}

Now we describe a hard instance of the special rectangle stabbing problem to establish
a lower bound. It will have rectangles of $h$ different shapes. 
For each shape, we tile (disjointly cover) the unit cube 
using isometric copies of the shape to obtain a set of rectangles.
We collect
every $r$ different shapes into a class and obtain $h/r$ classes, where $r$ is
a parameter to be determined later. 
We say that the $j$-th rectangle in a class has group number $j$, $j=0,1,\cdots, r-1$.
Now we specify the dimensions (i.e., side lengths) of the rectangles. 
For a rectangle in class $i=0,1,\cdots,h/r-1$, with group number $j=0,1,\cdots,r-1$, its dimensions are
\begin{align*}
  [\frac{1}{K^j}\times{K^j\cdot{2^{ir+j}}\cdot{V}}\times{\frac{1}{2^{ir+j}}}],
\end{align*}
where $K, V$ are parameters to be determined later and the $[W\times H \times D]$ notation denotes
an axis-aligned rectangle with width $W$, height $H$, and depth $D$.
Observe that every rectangle
has volume $V$ and thus we need $1/V$ copies to tile the unit cube. By setting $V=h/n$,
the total number of rectangles we generate is $n$. Also note that all the
rectangles in the same group are pairwise disjoint and they together cover the
whole unit cube. This implies for any query point $q$ in the unit cube, it is
contained in exactly $h=|\pi|$ rectangles.

Now we analyze the intersection of any two rectangles. 
First, observe that given two axis-aligned rectangles with dimensions 
$[W_1\times H_1 \times D_1]$ and $[W_2\times H_2 \times D_2]$, their intersection
is an axis-aligned rectangle with dimensions at most 
$[\min\left\{ W_1,W_2 \right\}\times\min\left\{ H_1,H_2 \right\}\times\min\left\{ D_1,D_2 \right\}]$.
Second, by our construction, the rectangles that have identical width, depth, and height are disjoint.
As a result, either the width of the two rectangles will differ by a factor $K$ or their
depth will differ by a factor $2^r$.
This means that, the maximum intersection volume of any two rectangles $R_1, R_2$
in class $i_1, i_2$, group $j_1, j_2$ can be achieved only in one of the following two
cases:

$$
V(R_1\cap{R_2})=
\begin{cases}
\frac{V}{2K} & i_1=i_2 \textrm{ and } j_1=j_2+1,\\
\frac{V}{2^r} & i_1=i_2+1 \textrm{ and } j_1=j_2.\\
\end{cases}
$$

We set $K=2^r$, then the intersection volume of any two rectangles is bounded
by $v=V/2^r$.
However, for the construction to be well-defined, the side length of the rectangles
cannot exceed 1 as otherwise, they do not fit in the unit cube.
The largest height of the rectangles is obtained for $j=r-1$ and $i=\frac{h}{r}-1$.
Thus, we  must have,
\begin{align*}
  K^{r-1}2^{r (\frac{h}{r}-1)+ r-1} V \le 1.
\end{align*}
By plugging in the values $V=h/n$ and $K=2^r$ we get that we must have
\begin{align}
  2^{r^2-r}2^{h-1} h \le n\label{eq:con1}
\end{align}
Since by our assumptions  $h \le \frac{\log n}{2}$, it follows that by 
setting $r= \frac{\sqrt{\log n}}{4}$, the inequality (\ref{eq:con1}) holds.

If $\gamma h \ge Q(n)$ holds, then we satisfy the first condition of \autoref{thm:framework} and thus
we obtain the space lower bound of
\begin{align}
  S(n)=\Omega\left(\frac{Q(n)v^{-1}}{2^{O(\gamma)}}\right)=\Omega\left(\frac{n2^r}{\log n 2^{O(\gamma)}}\right).
\end{align}
Now observe that if we set $\gamma = \delta \sqrt{\log n}$, for a sufficiently small $\delta> 0$, 
then it follows that the data structure must use more than
$\Omega(n 2^{\Omega(\sqrt{\log n})})$ space. 
However, by the statement of our lemma, we are not considering such data structures.
As a result, when $\gamma = \delta \sqrt{\log n}$, the query time must be large enough that the first condition of the framework 
does not hold, meaning, 
we must have $Q(n) \ge \gamma h = \delta \sqrt{\log n} h = \Omega(|\pi| \sqrt{\log n})$.  
\end{proof}

\ignore{
\begin{remark}
We will frequently use this idea to prove lower bounds. In the subsequent proofs, we will omit the description about the rectangle stabbing problem used explicitly since it should be clear from the construction of a hard input instance. We also omit the reduction from a rectangle stabbing problem to a 2D OFC problem as it is similar to the one we describe in the proof of \autoref{lem:shortpathtreelower}.
\end{remark}
}

\subsection{Trees of height $\frac{\log n}{2} < h \le\frac{\log^2{n}}{2}$}

\subsubsection{The Upper Bound}

We start with the following lemma which gives us a data structure that can only answer
query paths that start from the root and finish at a leaf. 
The main idea here is used previously in the context of four-dimensional dominance
queries~\cite{afshani2012higher,chan:revisit} and it uses the observation that
such ``root to leaf'' queries can be turned into a geometric problem, the 3D rectangle
stabbing problem.

\begin{lemma}
\label{lem:pathbalancedtreemiddle}
  Consider a balanced catalog tree of height $h$, $\frac{\log n}{2} < h \le\frac{\log^2{n}}{2}$,
  in which each vertex is associated with a planar subdivision.
  Let $n$ be the total complexity of the subdivisions. 
  We can build a data structure using $O(n)$ space such that given any query
  $(q, \pi)$, where $q$ is a query point and $\pi$ is a path starting from the root to a leaf, 
  all regions containing $q$ along $\pi$ can be reported in time
  $O(\sqrt{|\pi|}\log{n})$.  
\end{lemma}

\begin{proof}
Let $r$ be a parameter to be determined later. 
For each subdivision $A_i$, we create an intersection sensitive $(r/n_i)$-cutting $C_i$ on $A_i$.
By the same argument as \autoref{lem:shorttreepath}, 
all the cells in the cuttings are axis-aligned rectangles satisfying 
(i) the conflict set size of any cell in $C_i$ is bounded by $O(r)$ and 
(ii) the total number of cells in $C_i$ is $O(n_i/r)$.

Now we lift each cell in the cuttings to 3D rectangles and 
collect all the 3D rectangles to construct a 3D rectangle stabbing data structure for it. 
This is done as follows.
  We assign a $z$ range for each vertex in the catalog tree;
  Let $m$ be the number of leaves.
  Order the leaves of the catalog tree from left to right and for the
  $i$-th leaf $l_i, i\in\{1,2,\cdots,m\}$,
  we assign the range $[i-1,i)$ as its $z$ range. 
  For any internal vertex, its $z$ range is the union of the $z$ ranges of its children. 
  Then, we lift the 2D rectangles induced by the subdivision of a vertex to a 3D rectangle using the $z$ range
  (i.e., by forming the Cartesian product of the rectangle and the $z$ range).
  We store the 3D rectangles in a rectangle stabbing data structure.
  Given a query point $q=(x_q,y_q)$ and a query path $\pi$, 
  we first lift $q$ to be $(x_q,y_q,z_q)$, where $z_q$ is any $z$ value in the
  $z$ range of the deepest vertex in $\pi$, and then query the 3D rectangle stabbing data structure. 

In addition, for each cell in a cutting, we build an optimal point location
data structure on its conflict set. 
All these point location data structures take space $O(\sum_in_i)=O(n)$ in total 
and each of them can answer a point location query in time $O(\log{r})$.

To achieve space bound $O(n)$ for the 3D rectangle stabbing data structure, it suffices to choose $H=\frac{r}{\log{(n/r)}}$. 
We then balance the query time for 3D rectangle stabbing and 2D point locations to achieve the optimal query time

$$
\log n \cdot \frac{\log n }{\log \frac{r}{\log{(n/r)}}}=h \cdot \log r.
$$

 We pick $r=2^{\log n / \sqrt h}$ and the query time is bounded by $O(\sqrt{h}\log{n})=O(\sqrt{|\pi|}\log{n})$.
 \end{proof}

The above data structure is not a true fractional cascading data structure because it can only
support restricted queries. 
To be able to answer query paths of arbitrary lengths $>\frac{\log n}{2}$ and 
$\le\frac{\log^2 n}{2}$, we need the following result.

\begin{lemma}
\label{lem:pathtreetworanges}
  Consider a catalog tree in which each vertex is associated with a planar subdivision.
  Let $n$ be the total complexity of the subdivisions and 
  let $h_1$ and $h_2$, $h_1 < h_2$, be two fixed parameters. 
  We can build a data structure using $O(n\log(h_2/h_1))$ space such that given any query
  $(q, \pi)$, where $q$ is a query point and $\pi$ is a path whose length obeys $h_1 \le |\pi| \le h_2$,
  all regions containing $q$ along $\pi$ can be reported in time
  $O(\sqrt{|\pi|}\log{n})$.  
\end{lemma}
\begin{proof}
  First, observe that w.l.o.g., we can assume that the height of the catalog tree is at most
  $h_2$: we can partition the catalog tree into a forest by cutting off vertices whose depth is a 
  multiple of $h_2$. Since the length of $\pi$ is at most $h_2$, it follows that $\pi$ can only
  contain vertices from at most two of the trees in the resulting forest, meaning, answering
  $\pi$ can be reduced to answering at most two queries on trees of height at most $h_2$.

  Thus, w.l.o.g., assume $v$ is the root of the catalog tree of height $h_2$ and $\pi$ is a
  path of length at least $h_1$ in this catalog tree.
  We build the following data structures.
  Let $v_1, \cdots, v_m$ be the vertices at height $h_2/2$.
  Let $T_0$ be the tree rooted at $v$ and cut off at height $h/2$ with $v_1, \cdots, v_m$ 
  being leafs and $T_i$ be the tree rooted at $v_i$, $1 \le i\le m$.
  We build $m+1$ data structures of \autoref{lem:pathbalancedtreemiddle} on $T_0, \cdots, T_m$
  and then we recurse on each of the $m+1$ trees.
  The recursion stops once we reach subproblems on trees of height at most $h_1$.

Since the data structure of \autoref{lem:pathbalancedtreemiddle} uses $O(n)$ space, 
at each recursive level, the total space usage of data structures we constructed is $O(n)$.
Over the $O(\log(h_2/h_1))$ recursion levels, this sums up to $O(n\log(h_2/h_1))$ space. 

Now we analyze the query time. 
Given a query $(q, \pi)$, we may query several data structures that together cover the whole path of $\pi$. 
Let $w$ be the highest vertex on $\pi$. 
We can decompose $\pi$ into two disjoint parts $\pi_1$ and $\pi_2$,  that start from $w$ and
end at vertices $u_1$ and $u_2$ respectively, with $u_1$ and $u_2$ being descendants of $w$.
It thus suffices to only answer $\pi_1$, as the other path can be answered similarly. 
The first observation is that we can find a series of data structures that can be used
to answer disjoint parts of $\pi_1$.
The second observation is that we can afford to make the path a bit longer to truncate the recursion.
We now describe the details.

Consider the trees $T_0, \cdots, T_m$ defined at the top level of the recursion.
If $\pi_1$ is entirely contained in one of the trees, then we recurse on that tree.
Otherwise, $w$ is contained in $T_0$ and $u_1$ is contained in some subtree $T_i$.
Now, $\pi_1$ can be further subdivided into two smaller ``anchored'' paths: 
one from $w$ to $v_i$ (``anchored'' at $w$) 
and another from $v_i$ to $u_1$ (``anchored'' at $v_i$)
and each smaller path can be answered recursively in the corresponding tree.
Thus, it suffices to consider answering the query $q$ along an anchored path.

Thus, consider the case of answering an anchored path $\pi'$ in the data structure.
To reduce the notation and clutter, assume $\pi'$ is an anchored path, starting from the root of $T$ and ending
at a vertex $u$. 
Assume the vertices $v_1, \cdots, v_m$ and trees $T_0, \cdots, T_m$ are defined as above.
First, consider the case when the height of $T$ is at most $h_1$;
in this case, we have built an instance of the data structure of \autoref{lem:pathbalancedtreemiddle} on $T$
but not on the trees $T_0, \cdots, T_m$.
In this case, we simply answer $q$ on a root of leaf path in $T$ that includes $\pi'$, e.g., by picking
a leaf in the subtree of $u$. 
In this case, we will be performing a number of ``useless'' point location queries, in particular
those on the descendants of $u$.
However, as the height of $T$ is at most $h_1$, it follows that the query bound stays asymptotically
the same: $O(\sqrt{h_1} \log n)$. 
Furthermore, there is no recursion in this case and thus this cost is paid only once per anchored path.
The second case is when the height of $T$ is greater than $h_1$.
In this case, if $u$ lies in $T_0$ we simply recurse on $T_0$ but if $u$ lies in a tree
$T_i$, we first query the data structure of \autoref{lem:pathbalancedtreemiddle} using the path
from the root of $T$ until $v_i$, and then we recurse on $T_i$.
As a result, answering the anchored path query reduces to answering at most one query 
on an instance of data structure~\autoref{lem:pathbalancedtreemiddle} and another recursive ``anchored''
on a tree of half the height.
Thus, the $i$-th instance of the data structure~\autoref{lem:pathbalancedtreemiddle} that we
query covers at most $1/2^i$ fraction of the anchored path. 
Thus, if $k$ is the length of the anchored path, it follows that the
total query time of all the data structures we query is bounded by
$$
\sum_{i=1}^{\infty}\frac{\sqrt{k}\log n}{2^i}=O(\sqrt{k}\log n)=O(\sqrt{|\pi|}\log n).
$$
\end{proof}

\ignore{
We now reduce the space of the above lemma dramatically. 
Let $\log^* n$ be the iterated $\log$ function, i.e., the number of times
we need to apply $\log$ function to $n$ until we reach 2.
Let $\log^{*(i)}n$ be the application of $\log^*$ function $i$ times, i.e.,
$\log^{*(i)}n = \overbrace{\log^*(\log^*(\cdots \log^*(n)\cdots)}^{i} $.
Finally, define $\alpha_5(n)$ to be the value $i$ such that
$\log^{*(i)}n = 2$, i.e., how many times we need to apply the $\log^*$ function
until we reach a constant.
}

We now reduce the space of the above lemma dramatically. 
We will repeatedly use a ``bootstrapped" data structure. 
The following lemma establishes how we can bootstrap a base data structure 
to obtain a more efficient one.

\begin{lemma}
\label{lem:pathtreetworangesbootstrap}
  Consider a catalog tree of height $h$, $\frac{\log n}{2} < h \le\frac{\log^2{n}}{2}$, 
  in which each vertex is associated with a planar subdivision.
  Let $n$ be the total complexity of the subdivisions. 
  Assume, for any fixed value $\Delta$, $\omega(1) \le \Delta \le \log n$, we can build a ``base'' data structure 
  that can answer a 2D OFC query $(q,\pi)$ in $Q_b(n) = O(\sqrt{|\pi|}\log n)$ time as long as 
  $\pi$ is path of length between $\frac{\log^2 n}{2\Delta}$ and $\frac{\log^2 n}{2}$. 
  Furthermore, assume it uses $S_b(\Delta,n) = O(nf(\Delta))$ space, for some function $f$ which 
  is monotone increasing in $\Delta$ and for $\Delta = \omega(1)$ we have
  $f(\Delta) = \omega(1)$.

  Then, for any given fixed value $\Delta$, $\omega(1) \le \Delta \le \log n$, we can build a
  ``bootstrapped'' data structure 
  that can answer a 2D OFC query $(q,\pi)$ in $Q_b(n) +  O(\sqrt{|\pi|}\log n)$ time as long as 
  $\pi$ is path of length between $\frac{\log^2 n}{2\Delta}$ and $\frac{\log^2 n}{2}$. 
  Furthermore, it uses $O(n f^*(\Delta))$ space, where $f^*(\cdot)$ is the iterative $f(\cdot)$ function
  which denotes how many times we need to apply $f(\cdot)$ function to $\Delta$ to reach a constant value.
\end{lemma}

\begin{proof}
  We construct an intersection sensitive $(f(\Delta)/n_i)$-cutting $C_i$
  for each planar subdivision $A_i$ attached to the tree. 
  Call these the ``first level'' cuttings. 
  Similar to the analysis in \autoref{lem:shorttreepath}, we obtain $O(n_i/f(\Delta))$ cells, 
  which are disjoint axis-aligned rectangles, for each $C_i$ and thus $n'= O(n/f(\Delta))$ cells in total.
  Each cell in the cutting has a conflict list of size $O(f(\Delta))$ and on that we build a point
  location data structure.
  This takes $O(n)$ space in total.
  We store the cells of the cutting in an instance of the base data structure with parameter $\Delta$.
  Call this data structure $\A_1$.
  The space usage of $\A_1$ is 
  \begin{align*}
    S_b(\Delta, n') = O(n'f(\Delta)) = O(n).
  \end{align*}

  Now we consider a query $(q,\pi)$.
  Let $\delta_1= \log(f(\Delta))$.
  Consider the case when $\frac{1}{2}\cdot\left( \frac{\log^2 n}{\Delta} \right) \le
  |\pi| \le  \frac{1}{2}\cdot\left( \frac{\log n}{\delta_1} \right)^2$.
  In this case, as $\A_1$ is built with parameter $\Delta$, we can query it with $(q,\pi)$.  
  Thus, in $Q_b(n)$ time, for every subdivision on path $\pi$, we find the cell of the cutting that contains
  $q$.
  Then, we use the point location data structure on the conflict lists of the
  cells to find the original rectangle containing
  $q$.
  This takes an additional $O(\log(f(\Delta)))$ as the size of each conflict is $O(f(\Delta))$.
  Thus, the query time in this case is
  \begin{align*}
    Q_b(n) + O(|\pi|\log(f(\Delta))) = Q_b(n) + O(\sqrt{|\pi|}\log n)
  \end{align*}
  since we have  $|\pi|\le  \frac{1}{2}\cdot\left( \frac{\log n}{\delta_1} \right)^2 = \frac{1}{2}\cdot \left( \frac{\log n}{\log(f(\Delta))} \right)^2$.

  Thus, the only paths we cannot answer yet are those when 
  $\frac{1}{2}\cdot\left( \frac{\log n}{\delta_1} \right)^2 \le |\pi| \le \frac{\log^2 n}{2}$.
  In this case, we can bootstrap. 
  First, observe that we can build a data structure $\A'$ on the 
  the original rectangles, where $\A'$ is an instance of the base data
  structure but this time with  parameter $\Delta$ set to  $\delta_1^2$.
  \ignore{
  This will take $S_b(\delta_1^2,n) = O(n f(\delta_1^2,n))$ space. 
  Thus, the total space consumption is
  \begin{align}
    O(n) + S_b(\delta_1^2, n) = O(n) +  O(n f[ \log^2(f(\Delta)),n]) = O(n) + O(f(f(\Delta)))\label{eq:space2}
  \end{align}
  }
  This will take $S_b(\delta_1^2,n) = O(n f(\delta_1^2))$ space. 
  Thus, the total space consumption is
  \begin{align}
    O(n) + S_b(\delta_1^2, n) = O(n) +  O(n f(\log^2(f(\Delta)))) = O(n) + O(nf(f(\Delta)))\label{eq:space2}
  \end{align}
  where the last inequality follows since $f(\cdot)$ is a monotone increasing function and 
  $\log^2(f(\Delta)) < f(\Delta)$ as $f(\Delta) = \omega(1)$.
  By construction, the data structure $\A'$ is built to handle exactly paths of this length
  but it is using too much space.
  The idea here is that we can repeat the 
  previous technique using ``second level'' cuttings to obtain a data structure $\A_2$:
  for a subdivision of size $n_i$, build a $(f(f(\Delta))/n_i)$-cutting, called the ``second level'' cutting.
  By repeating the same idea we used for the first level cuttings, we can spend additional $O(n)$ space
  to build a data structure $\A_2$ which can answer queries $(q,\pi)$ as long as 
  $\frac{1}{2}\cdot \left( \frac{\log^2 n}{\Delta} \right) \le
  |\pi| \le  \frac{1}{2}\cdot \left( \frac{\log n}{\delta_2} \right)^2$
  where $\delta_2 = \log(f(f(\Delta)))$.
  By repeating this process for $f^*(\Delta)$ steps, we can obtain the claim data structure. 
\end{proof}

We will essentially begin with the data structure in \autoref{lem:pathtreetworanges} 
and use \autoref{lem:pathtreetworangesbootstrap} to bootstrap.
To facilitate the description, we define two useful functions first.
Let $\log^* n$ be the iterated $\log$ function, i.e., the number of times
we need to apply $\log$ function to $n$ until we reach 2.
We define $\log^{*(i)}n$ as follows
$$
\log^{*(i)}n=
\begin{cases}
\log^{*(i)}(\log^{*(i-1)}n)+1 & n>1,\\
0 & n\le1.\\
\end{cases}
$$
In other words, it is the number of times we need to apply $\log^{*(i-1)}$ function to $n$ until we reach 2.
We also defined the following function
$$
\tau(n)=\{\min i: \log^{*(i)}n \le 3\}.
$$

In fact, $\log^{*(i)}n$ is the $(i+2)$-th function of the inverse Ackermann hierarchy~\cite{cormen2009introduction} 
and $\tau(n)=\alpha(n)-3$, where $\alpha(n)$ the inverse Ackermann function~\cite{cormen2009introduction}.

  \begin{lemma}
  \label{lem:pathtree}
    Consider a catalog tree of height $h$, $\frac{\log n}{2} < h \le\frac{\log^2{n}}{2}$,
    in which each vertex is associated with a planar subdivision.
    Let $n$ be the total complexity of the subdivisions. 
    We can build a data structure using $O(n\alpha_c(n))$ space, where $c\geq{3}$ is any constant 
    and $\alpha_c(n)$ is the $c$-th function of the inverse Ackermann hierarchy,
    such that given any query $(q, \pi)$, where $q$ is a query point and $\pi$ is a path 
    of length $|\pi|$, $\frac{\log n}{2} < |\pi| \le\frac{\log^2{n}}{2}$,
    all regions containing $q$ along $\pi$ can be reported in time $O(\sqrt{|\pi|}\log{n})$.
    Furthermore, we can also build a data structure using $O(n)$ space 
    answering queries in time $O(\alpha(n)\sqrt{|\pi|}\log{n})$, 
    where $\alpha(n)$ is the inverse Ackermann function.
  \end{lemma}

\begin{proof}
By \autoref{lem:pathtreetworanges}, if we set $h_1=\frac{\log n}{2}$ and $h_2=\frac{\log^2 n}{2}$, 
we obtain a data structure using $O(n\log\log{n})$ answering queries in time $O(\sqrt{|\pi|}\log n)$.
By picking $\Delta=\log n$, $f=\log n$, we can apply \autoref{lem:pathtreetworangesbootstrap}
to reduce the space to $O(n\log^*(\log n))=O(n\log^*n)$ while achieving the same query time.
If we again pick $\Delta=\log n$, but $f=\log^* n$, by applying \autoref{lem:pathtreetworangesbootstrap} again,
the space is further reduced to $O(n\log^{**}(\log n))=O(n\log^{**}n)$.
We continue this process until $\log^{*(i)}n$ is less than three. Note that we will need to pay 
$O(\sqrt{|\pi|}\log{n})$ extra query time each time we apply \autoref{lem:pathtreetworangesbootstrap}.
We will end up with a linear-sized data structure with query time $O(\tau(n)\sqrt{|\pi|}\log{n})=O(\alpha(n)\sqrt{|\pi|}\log{n})$.
On the other hand, if we stop applying \autoref{lem:pathtreetworangesbootstrap} after a constant $c$ many rounds, 
we will end up with a $O(n\log^{*(c)}n)=O(n\alpha_{c+2}(n))$ sized data structure 
with the original $O(\sqrt{|\pi|}\log{n})$ query time.
\end{proof}

\subsubsection{The Lower Bound}
We show an almost matching lower bound in this section.

\begin{lemma}
\label{lem:pathtreelower}
Assume, given any catalog tree of height $h$, $\frac{\log n}{2} < h \le\frac{\log^2{n}}{2}$,
in which each vertex is associated with a planar subdivision
  with $n$ being the total complexity of the subdivisions, we can build a data structure that satisfies the following:
  it uses at most $n 2^{\varepsilon \log n / \sqrt h}$ space, for a small enough constant $\varepsilon$, and 
  it can answer 2D OFC queries $(q,\pi)$. 
  Then, its query time must be $\Omega(\sqrt{|\pi|}\log{n})$.
\end{lemma}

\begin{proof}
  We first describe a hard input instance for a 3D rectangle stabbing problem and
  later we show that this can be embedded as an instance of 2D OFC problem on
  a tree of height $h$.
  Also, we actually describe a tree of height $h' = \frac{h + (\log n)/4}{2} \le h$. 
  This is not an issue as we can add dummy vertices to the root to get the height to exactly $h$.

  We begin by describing the set of rectangles. 
  Each rectangle is assigned a ``class number'' and a ``group number''. 
  The number of classes is $\frac{\sqrt{h}}{2}$ and the number of groups is 
  $\frac{\log{n}}{4\sqrt{h}}+\sqrt{h}$.
  The rectangles with the same class number and group number will be disjoint, isometric and 
  they would tile the unit cube. 
Rectangles with class number $i=0,1,\cdots,\frac{\sqrt{h}}{2} - 1$ and 
group number $j=0,1,\cdots,\frac{\log{n}}{4\sqrt{h}}-1$ will be of shape
$$
[\frac{1}{K^j}\times{K^j}\cdot{2^{ir+j}}\cdot{V}\times{\frac{1}{2^{ir+j}}}],
$$
where $V, K$ are some parameters to be determined later and $r=\frac{\log n}{4\sqrt{h}}$.
Similarly, rectangles with class number $i=0,1,\cdots,\frac{\sqrt{h}}{2} -1$ and group number
$j=\frac{\log{n}}{4\sqrt{h}},\cdots,\frac{\log{n}}{4\sqrt{h}}+\sqrt{h}-1$ will be of shape
$$
[\frac{1}{K^j}\times{K^j}\cdot{2^{ir+r}}\cdot{V}\times{\frac{1}{2^{ir+r}}}].
$$
The total number of different shapes is $\frac{\sqrt{h}}{2}\cdot ( \frac{\log{n}}{4\sqrt{h}}+\sqrt{h}) = h'$.
Note that each rectangle has volume $V$, so the total number of rectangles we use in all the tilings is $n$ by setting $V=h'/n$. 
By our construction any query point is contained in $t=h'=|\pi|$ rectangles. 
Now we analyze the maximal intersection volume of two rectangles. 
By the same argument as in the proof of \autoref{lem:shortpathtreelower} 
the maximal intersection volume can only be achieved by two rectangles when they are 
in the same class and adjacent groups or in the same group of adjacent classes. 
For two rectangles $R_1$ and $R_2$ in group $j_1, j_2$ of class $i_1, i_2$, we have

$$
V(R_1\cap{R_2})=
\begin{cases}
\frac{V}{2K} & i_1=i_2 \textrm{ and } j_1=j_2+1\leq{\frac{\log{n}}{4\sqrt{h}}},\\
\frac{V}{K} & i_1=i_2 \textrm{ and }\frac{\log{n}}{4\sqrt{h}}\leq{j_1}=j_2+1,\\
\frac{V}{2^r} & i_1=i_2+1 \textrm{ and } j_1=j_2.\\
\end{cases}
$$

We set $K=2^r$, then the intersection of any two rectangle is no more than $v=V/2^r$. 

We also need to make sure no side length of any rectangle exceeds the side length of the unit cube. 
The maximum side length can only be obtained 
when $i=\frac{\sqrt{h}}{2} - 1$ and $j=\frac{\log{n}}{4\sqrt{h}}+\sqrt{h}-1$ in the second dimension. 
We must have
$$
K^{\frac{\log{n}}{4\sqrt{h}}+\sqrt{h}-1} \cdot 2^{r(\frac{\sqrt{h}}{2} - 1)+\frac{\log{n}}{4\sqrt{h}}} \cdot V \le 1
$$
Plugging $K=2^r$ and $V=h'/n$ in, we must have
\begin{align}
2^{r\frac{\log{n}}{4\sqrt{h}}+r(\sqrt{h}-1)} \cdot 2^{r(\frac{\sqrt{h}}{2} - 1)+\frac{\log{n}}{4\sqrt{h}}} \cdot h' \le n \label{eq:con4.3}
\end{align}
Since $\frac{\log n}{2} \le h \le \frac{\log^2 n}{2}$ and $r=\frac{\log{n}}{4\sqrt{h}}$, (\ref{eq:con4.3}) holds.

Suppose $\gamma h \ge Q(n)$, then the first condition of \autoref{thm:framework} is satisfied and 
we get the lower bound of
$$
S(n)=\Omega\left(\frac{tv^{-1}}{2^{O(\gamma)}}\right)=\Omega\left(\frac{n2^r}{2^{O(\gamma)}}\right).
$$

Observe that by setting $\gamma=\frac{\delta\log n}{\sqrt h}$ for a sufficiently small $\delta>0$, 
the data structure must use $\Omega(n2^{\Omega(\log n / \sqrt h)})$ space, which contradicts the space usage in our theorem. 
Therefore, $Q(n) \ge \gamma h = \frac{\delta\log n}{\sqrt h} h=\Omega(\sqrt{h}\log{n})=\Omega(\sqrt{|\pi|}\log{n})$.
\begin{figure}[h]
  \centering
  \includegraphics[scale=1]{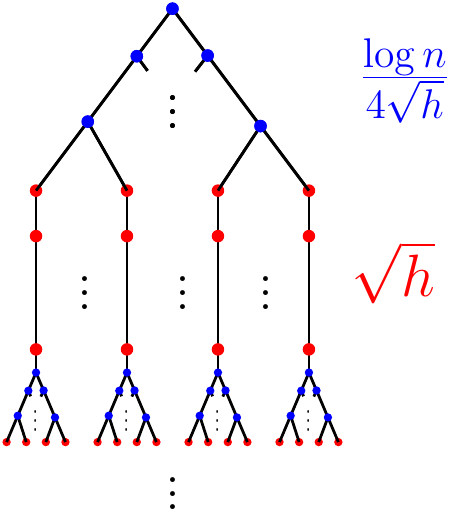}
  \caption{A difficult tree for fractional cascading.}
  \label{fig:lbtree}
\end{figure}
It remains to show that this set of rectangles can actually be embedded into an instance of the 2D OFC problem.
To do that, we describe the tree $T$ that can be used for this embedding. 
See Figure~\ref{fig:lbtree}.
We hold the convention that the root of $T$ has depth $0$.
Starting from the root, until depth $\frac{\log n}{4\sqrt{h}}$, every vertex will have two children
(blue vertices in Figure~\ref{fig:lbtree}) then we will have $\sqrt{h}$ vertices with one child
(red vertices in Figure~\ref{fig:lbtree}).
Then this pattern continues for $\frac{\sqrt{h}}{2}$ steps.
The first set of blue and red vertices correspond to class $0$, the next to class $1$ and so on. 
Within each class, the top level corresponds to group $0$ and so on.
To be specific, vertices at depth $(\frac{\log{n}}{4\sqrt{h}} + \sqrt{h})i + j$ of the tree 
have rectangles of class $i$ and group $j$. 
Now, it can be seen that the rectangles can be assigned to the vertices of $T$, 
similar to how it was done in \autoref{lem:shortpathtreelower}.
The notable difference here is that the depth (length of the side parallel to the $z$-axis) 
of the rectangles decreases as the group number
increases from $0$ to $\frac{\log{n}}{4\sqrt{h}}-1$ but then it stays the same from 
$\frac{\log{n}}{4\sqrt{h}}$ until $\frac{\log{n}}{4\sqrt{h}}+\sqrt{h}-1$.
This exactly corresponds to the structure of the tree $T$.
\end{proof}

\subsection{Trees of height $h >\frac{\log^{2}n}{2}$}

For trees of this height, we have: 

\begin{lemma}
\label{lem:pathtreelong}
  Consider a catalog tree of height $h > \frac{\log^2n}{2}$ in which each vertex is associated with a planar subdivision.
  Let $n$ be the total complexity of the subdivisions. 
  We can build a data structure using $O(n)$ space such that given any query
  $(q, \pi)$, where $q$ is a query point and $\pi$ is a path of length $|\pi|>\frac{\log^2n}{2}$, 
  all regions containing $q$ along $\pi$ can be reported in time $O(|\pi|)$.  
\end{lemma}

\begin{proof}
We combine the classical heavy path decomposition by Sleator and Tarjan \cite{sleator1983data} and the data structure for catalog paths to achieve the desire query time. We first apply the heavy path decomposition to the tree and then for every heavy path created we build a 2D OFC data structure to answer queries along the path. Clearly, we only spend linear space in total. Then by the property of the heavy path decomposition, we only need to query $O(\log{n})$ heavy paths to answer a query, which leads to a query time of $O(\log^2{n}+|\pi|)=O(|\pi|)$.
\end{proof}


By combining \autoref{lem:shorttreepath}, \autoref{lem:pathtree}, \autoref{lem:pathtreelong}, 
we immediately get the following corollary.

\begin{corollary}
\label{cor:pathtree}
    Consider a catalog tree in which each vertex is associated with a planar subdivision.
    Let $n$ be the total complexity of the subdivisions. 
    We can build a data structure using $O(n\alpha_c(n))$ space, where $c\geq{3}$ is any constant 
    and $\alpha_c(n)$ is the $c$-th function of the inverse Ackermann hierarchy,
    such that given any query $(q, \pi)$, where $q$ is a query point and $\pi$ is a path, 
    all regions containing $q$ along $\pi$ can be reported in time 
    $O(\log{n}+|\pi|+\min\{|\pi|\sqrt{\log{n}},\sqrt{|\pi|}\log{n}\})$.
    Furthermore, we can also build a data structure using $O(n)$ space 
    answering queries in time $O(\log{n}+|\pi|+\min\{|\pi|\sqrt{\log{n}},\alpha(n)\sqrt{|\pi|}\log{n}\})$, 
    where $\alpha(n)$ is the inverse Ackermann function.
\end{corollary}

\section{Queries on Catalog Graphs and Subgraph Queries}

In this section, we consider general catalog graphs as well as subgraph 
queries on catalog trees.
Our result shows that it is possible to build a data structure of space
$O(n)$ such that we can save a $\sqrt{\log{n}}$ factor from the naïve query
time of iterative point locations. We also present a matching lower bound.

We begin by presenting a basic reduction. 
\begin{lemma}
\label{lem:graphtrans}
  Given a catalog graph $G$ of $m$ vertices with graph subdivision complexity $n$
  and maximum degree $d$, we can generate a new catalog graph $G'$ with 
  $\Theta(md)$ vertices with graph subdivision complexity $\Theta(dn)$ and bounded
  degree $O(d^2)$ such that the following holds: given any connected subgraph $\pi\subseteq{G}$, 
  in time $O(|\pi|)$, we can find a path $\pi'$ in $G'$  such that the answer to any query
  $Q_1=(q,\pi)$ in $G$ equals the answer to query $Q_2=(q, \pi')$ in $G'$.
\end{lemma}

\begin{proof}
    The main idea is that we can add a number of dummy vertices to the graph such that 
    we can turn a subgraph query to a path query.

    We can obtain $G'$ in the following way.
    For every vertex in $G$, place $2d$ copies of the vertex in $G'$.
    All the copies of a vertex are connected in $G'$.
    Furthermore, every copy of a vertex $v_i$ is connected to every copy of vertex $v_j$
    if and only if $v_i$ and $v_j$ are connected in $G$.
    The maximum degree of $G'$ is thus $O(d^2)$.

    Now consider a subraph query $\pi$ in $G$. 
    By definition, $\pi$ is a connected subgraph of $G$ and w.l.o.g., we can assume
    $\pi$ is a tree. 
    We can form a walk $W$ from $\pi$ by following a DFS ordering of $\pi$ such that 
    $W$ traverses every edge of $\pi$ at most twice and visits every vertex of $\pi$.
    We observe that we can realize $W$ as a path in $G'$ by utilizing the dummy vertices;
    as each vertex has $2d$ dummy vertices, every visit to a vertex in $G$ can be replaced by a 
    visit to a distinct dummy vertex. 
\end{proof}

\subsection{The Upper Bound}

Formally, we have the following result.

\begin{lemma}
\label{thm:pathgraph}
  Consider a degree-bounded catalog graph in which each vertex is associated with a planar subdivision.
  Let $n$ be the total complexity of the subdivisions. 
  We can build a data structure using $O(n)$ space such that given any query
  $(q, \pi)$, where $q$ is a query point and $\pi$ is a path, all regions
  containing $q$ along $\pi$ can be reported in time
  $O(\log{n}+|\pi|\sqrt{\log{n}})$. 
\end{lemma}

\begin{proof}
We build the data structure essentially the same way as in the proof of \autoref{lem:shorttreepath}. 
The only difference is that the degree of a node is bounded by a $d\geq{2}$ which can be any constant. 
By the same argument in the proof of \autoref{lem:shorttreepath}, 
any node is stored at most $\Theta(r^{\log{d}}\log{r})$ times and we can obtain a $O(n)$ space bounded data structure achieving $O(\log{n}+|\pi|\sqrt{\log{n}})$ query time by creating an intersection sensitive $(r^{2\log{d}}/n_i)$-cutting for each planar subdivision $A_i$ and balancing the query time of cutting cells and conflict lists.
\end{proof}

By \autoref{lem:graphtrans}, we can also obtain the following two corollaries.
\begin{corollary}\label{cor:subgraph}
  Consider a catalog graph in which each vertex is associated with a planar subdivision.
  Let $n$ be the total complexity of the subdivisions. 
  We can build a data structure using $O(n)$ space such that given any query
  $(q, \pi)$, where $q$ is a query point and $\pi$ is a connected subgraph, 
  all regions containing $q$ along $\pi$ can be reported in time $O(\log{n}+|\pi|\sqrt{\log{n}})$.  
\end{corollary}

Specifically, for catalog trees we have the following:

\begin{corollary}\label{cor:subtree}
  Consider a catalog tree in which each vertex is associated with a planar subdivision.
  Let $n$ be the total complexity of the subdivisions. 
  We can build a data structure using $O(n)$ space such that given any query
  $(q, \pi)$, where $q$ is a query point and $\pi$ is a subtree, 
  all regions containing $q$ along $\pi$ can be reported in time $O(\log{n}+|\pi|\sqrt{\log{n}})$.  
\end{corollary}

\subsection{The Lower Bound}
\label{sec:pathgraphlower}
In this section, we show that the $\sqrt{\log n}$ factor that exists in \autoref{thm:pathgraph}, \autoref{cor:subgraph}, and
\autoref{cor:subtree} is tight. 
Like the proof of path queries for catalog trees, we need a reduction from a
rectangle stabbing problem to a 2D OFC subtree query problem on catalog trees.
But unlike previous proofs, we use an instance of the rectangle stabbing problem in a much higher dimension.
We show the lower bound for subtree queries of a catalog tree. 
By \autoref{lem:graphtrans}, this also gives a lower bound for 
path queries in general catalog graphs.

\begin{lemma}
\label{thm:treelowerbound}
Assume, given any catalog tree of height $\sqrt{\log n} \le h\le \frac{\log n}{2}$ 
in which each vertex is associated with a planar subdivision
  with $n$ being the total complexity of the subdivisions, 
  we can build a data structure that satisfies the following:
  it uses at most $n 2^{\varepsilon \sqrt{\log n}}$ space, for a small enough constant $\varepsilon$, and 
  it can answer 2D OFC queries $(q,\pi)$, where $q$ is a query point and 
  $\pi$ is a subtree containing $b=2^{h/2}$ leaves. 
  Then, its query time must be $\Omega(|\pi|\sqrt{\log{n}})$.
\end{lemma}

\begin{proof}

We define the following special $(2+b)$-dimensional rectangle stabbing problem. 
The input consists of $n$ rectangles in $(2+b)$ dimensions.
According to their shapes, rectangles are divided into $h'=h-r$ sets of size $n/h'$ each
where $r$ is a parameter to be determined later. 
The rectangles in each set are further divided into $b$ groups of size $n/(h'b)$ each.
All the rectangles in the same group are pairwise disjoint and they together tile the $(2+b)$-dimensional unit cube.
We put restrictions on the shapes of the input rectangles to make this problem special.
For a rectangle in set $i$, $i=0,1,\cdots,h'-1$, group $j$, $j=0,1,\cdots,b-1$,
except for the first two and the $(2+j)$-th dimensions,
its other side lengths are all set to be $1$.
The side length of the $(2+j)$-th dimension is set to be $1/2^{i+r}$.
We put restriction on the side lengths of the first two dimensions in set $i$ as follows:
For the first group $j=0$ in this set,
we put no restrictions of the first two dimensions as long as
they tile the unit cube and the total number of rectangles used for this group is $n/(h'b)$.
For an arbitrary group $j$,
we cut the range of the unit cube in the $(2+j)$-th dimension into $2^{i+r}$ equal length pieces.
This partitions the unit cube into $2^{i+r}$ parts.
Note that each part of the unit cube is also tiled by rectangles
since we require the side length of the $(2+j)$-th dimension of the rectangles to be $1/2^{i+r}$.
If we project the rectangles in each part of the unit cube into the first two dimensions,
we obtain $2^{i+r}$ axis-aligned planar subdivisions.
The planar subdivisions we generated for the first group is used
as a blueprint for the shape of rectangles in other groups.
More specifically, for the remaining groups in set $i$,
we require that the choices of the first two dimensions give
the same set of $2^{i+r}$ planar subdivisions as the first group.
The problem is as follows: Given a point in $(2+b)$-dimensions,
find all the rectangles containing this query point.

We now describe a reduction from this problem to a 2D OFC subtree query problem on catalog trees. 
We consider a complete balanced binary tree of height $h=h'+r$.
Note that the number of nodes at layer $i+r$ of the tree is the same as the number of different subdivisions 
we get by projecting a group in set $i$ to the first two dimensions.
Since we require that all the groups in the same set yield the same set of
subdivisions, we can simply attach the subdivisions to the nodes starting from layer $r$.
For nodes in layer smaller than $r$, we attach them with empty subdivisions.
Now let us analyze a rectangle stabbing query $q$ on the rectangle stabbing problem.
Consider rectangles in set $i$, we need to find the rectangle containing $q$ in each of the $b$ groups.
In this special rectangle stabbing problem,
to find the rectangle in group $j$ containing $q$, 
we can find the rectangle by first using the $(2+j)$-th coordinate of $q$ to
find the part of the unit cube where $q$ is in,
and then finding the output rectangle by a simple planar point location
on the projection of the part using the first two coordinates of $q$.
By our construction this is equivalent to choosing a node in layer $i+r$ of the binary tree
and performing a point location query on the subdivision attached to it.
Note that the node in layer $i+r+1$ we choose must be one of the children of the chosen node in layer $i+r$.
So if we only focus on one specific group $j$ of all sets,
the rectangle stabbing query corresponds to a series of point location queries
from the root to a leaf in the binary tree we constructed.
Similarly, we obtain $b$ such paths if we consider all groups and they together form a subtree of $b$ leaves.
The answer to the point location queries along the subtree gives the answer to the rectangle stabbing problem.

We describe a hard high dimensional rectangle stabbing problem instance. 
As before, we create rectangles of different shapes to tile the unit cube. 
But this time, we will consider a $(2+b)$-dimensional rectangle stabbing problem. 
For rectangles in class $i=0,\cdots,h'/r-1$ , supercluster $j=0,\cdots,r-1$, 
we create the following shapes:

\begin{align*}
\begin{rcases}
\lbrack\frac{1}{K^j}\times{K^j}\cdot{2^{ir+j+r}}\cdot{V}&\times{\frac{1}{2^{ir+j+r}}}\times{1}\times{1}\times\cdots\times{1}\rbrack\\
\lbrack\frac{1}{K^j}\times{K^j}\cdot{2^{ir+j+r}}\cdot{V}&\times{1}\times{\frac{1}{2^{ir+j+r}}}\times{1}\times\cdots\times{1}\rbrack\\
&\vdots\\
\lbrack\frac{1}{K^j}\times{K^j}\cdot{2^{ir+j+r}}\cdot{V}&\times{1}\times{1}\times\cdots\times{1}\times{\frac{1}{2^{ir+j+r}}}\rbrack\\
\end{rcases}
b\text{ shapes}
\end{align*}

where $K, V$ are parameters to be determined later. Note that all the rectangles are in $(2+b)$ dimensions.

We use each of the shape to tile a unit cube. Since the volume of any rectangle is $V$, we need $1/V$ rectangles of the same shape to tile the cube. We call it a cluster. Note that the rectangles in the same cluster are pairwise disjoint. We generated $h'b/V$ rectangles in total. By setting $V=h'b/n$, the total number of rectangles is $n$. Note that any point in the unit cube is contained in exactly $t=h'b$ rectangles.

Now we shall analyze the volume of the intersection between any two $(2+b)$-dimensional rectangles. 
Note that if two rectangles have the same side lengths for $b-1$ out of the last $b$ dimensions, then it is the case we have analyzed in the proof of, e.g., \autoref{lem:shortpathtreelower}, and the volume of the intersection of any two rectangles is bounded by $V/K$ if we set $K=2^r$. Now we analyze the other case. By our construction, two rectangles can only have at most two different side lengths in the last $b$ dimensions. We consider two rectangles in class $i_1$ supercluster $j_1$, and class $i_2$ supercluster $j_2$ respectively. W.l.o.g., we assume $j_1\geq{j_2}$. The case for $j_1\leq{j_2}$ is symmetric. Then there are two possible expressions for the intersection volume depending on the values of $i_1$ and $i_2$. The first one is
$$
\frac{1}{K^{j_1}}\times{K^{j_1}}\cdot{2^{i_1r+j_1+r}}\cdot{V}\times{\frac{1}{2^{i_1r+j_1+r}}}\times{\frac{1}{2^{i_2r+j_2+r}}}=\frac{V}{2^{i_2r+j_2+r}}\leq{\frac{V}{K}}.
$$
The second possible expression is
$$
\frac{1}{K^{j_1}}\times{K^{j_2}}\cdot{2^{i_2r+j_2+r}}\cdot{V}\times{\frac{1}{2^{i_1r+j_1+r}}}\times{\frac{1}{2^{i_2r+j_2+r}}}=\frac{V}{K^{j_1-j_2}}\times{\frac{1}{2^{i_1r+j_1+r}}}\leq\frac{V}{K}.
$$
The last inequality holds because $j_1\geq{j_2}$.

To make this construction well-defined, no side length of the rectangles can exceed 1. 
The largest side length can only be obtained in the second  dimension when $i=h'/r-1$ and $j=r-1$. We must have
\begin{align*}
  K^{r-1}2^{h'+r-1} V \le 1.
\end{align*}
By plugging in the values $V=h'b/n$ and $K=2^r$ we get that we must have
\begin{align}
  2^{r^2-r}2^{h'+r-1} h'b \le n\label{eq:con5.1}
\end{align}
Since by our assumptions  $h' \le \frac{\log n}{2}$, $b \le 2^{h/2} \le 2^{\log n / 4}$, 
it follows that by setting $r= \frac{\sqrt{\log n}}{4}$, the inequality (\ref{eq:con5.1}) holds.

If $\gamma h'b \ge Q(n)$ holds, then the first condition of \autoref{thm:framework} is satisfied and we obtain the lower bound of
$$
S(n)=\Omega\left(\frac{tv^{-1}}{2^{O(\gamma)}}\right)=\Omega\left(\frac{n2^r}{2^{O(\gamma)}}\right).\label{eq:con5.2}
$$
Now if we set $\gamma=\delta\sqrt{\log{n}}$ for a sufficiently small $\delta>0$, the data structure must use $\Omega(n2^{\Omega(\sqrt{\log{n}})})$ space, which contradicts the space usage stated in our lemma. 
Note that $|\pi|=\Theta((h'+r)b)=\Theta(h'b)$.
Then  $Q(n) \ge \gamma{h'b}=\Omega(|\pi|\sqrt{\log{n}})$.
\end{proof}

\begin{remark}
Note that the lower bound holds even when the query path is of length $\ge \sqrt{\log n}$ and $\le\frac{\log{n}}{2}$. We have already established this lower bound in \autoref{lem:shortpathtreelower}.
\end{remark}

Combining \autoref{lem:graphtrans} and \autoref{thm:treelowerbound}, we immediately have the following corollary:

\begin{corollary}
\label{cor:pathgraphlower}
Assume, given any degree-bounded catalog graph in which each vertex is associated with a planar subdivision
  with $n$ being the total complexity of the subdivisions, we can build a data structure that satisfies the following:
  it uses at most $n 2^{\varepsilon \sqrt{\log n}}$ space, for a small enough constant $\varepsilon$, and 
  it can answer 2D OFC queries $(q,\pi)$, where $q$ is a query point and $\pi$ is a path. 
  Then, its query time must be $\Omega(|\pi|\sqrt{\log{n}})$.
\end{corollary}

\section{Open Problems}
For the linear space data structure we obtained for general path queries of trees \autoref{cor:pathtree}, there is a tiny inverse Ackermann gap between the query time we obtain and the lower bound. It is an interesting problem whether we can get rid of that term or improve the lower bound. 

The problem we consider is very general in the sense that the only restriction we place on the input instance is that the graph subdivision complexity is $n$. Some special cases admit better solutions. For example, if we require the subdivision complexity of each vertex of the graph to be asymptotically the same, we can obtain an $O(n)$ space and $O(\log n + |\pi|\log\log{n})$ query time data structure\footnote{This is done by creating an intersection sensitive $(\log{n}/n_i)$-cutting $C_i$ for each subdivision $A_i$ in the tree and then storing all cutting cells on each path using the data structure in \autoref{thm:path} and building point location data structures on the conflict list of each cell.}
 for path queries on catalog trees of height $h\le\frac{\log{n}}{2}$, which is better compared to the linear space $O(\log n+|\pi|\sqrt{\log n})$ query time data structure we obtained in the general case \autoref{lem:shorttreepath}.

Higher dimensional generalization of our results is another direction. In 2D, we can transform an axis-aligned planar subdivision to a subdivision consisting of only rectangles by increasing the subdivision complexity by only a constant factor; however it is not the case for 3D. On the other hand, for 3D point locations on orthogonal subdivisions, we have Rahul's $O(\log^{3/2}n)$ query time and linear space data structure \cite{rahul2014improved} in the standard pointer machine model. Recently, the query time is improved to $O(\log{n})$ by Chan et al. \cite{chan2018orthogonal}, but they use a stronger arithmetic pointer machine model. Given that the higher dimensional counterparts of the tools we use for 2D are suboptimal, it is a challenging and interesting problem to see how the results will be in higher dimensions.

Other open problems include considering the dynamization of our results, i.e., to support insertion and deletion, and other computational models, e.g., RAM and I/O model.

\bibliography{reference}{}

\begin{thebibliography}{10}

\bibitem{afshani2012improved}
P.~Afshani.
\newblock Improved pointer machine and {I/O} lower bounds for simplex range
  reporting and related problems.
\newblock In {\em Proceedings of the twenty-eighth annual Symposium on
  Computational Geometry}, pages 339--346, 2012.

\bibitem{aal10}
P.~Afshani, L.~Arge, and K.~D. Larsen.
\newblock Orthogonal range reporting: query lower bounds, optimal structures in
  3-d, and higher-dimensional improvements.
\newblock In {\em Proceedings of the twenty-sixth annual symposium on
  Computational geometry}, pages 240--246, 2010.

\bibitem{afshani2012higher}
P.~Afshani, L.~Arge, and K.~G. Larsen.
\newblock Higher-dimensional orthogonal range reporting and rectangle stabbing
  in the pointer machine model.
\newblock In {\em Proceedings of the twenty-eighth annual Symposium on
  Computational Geometry}, pages 323--332, 2012.

\bibitem{afshani2014deterministic}
P.~Afshani, T.~M. Chan, and K.~Tsakalidis.
\newblock Deterministic rectangle enclosure and offline dominance reporting on
  the ram.
\newblock In {\em International Colloquium on Automata, Languages, and
  Programming}, pages 77--88. Springer, 2014.

\bibitem{afshani2014concurrent}
P.~Afshani, C.~Sheng, Y.~Tao, and B.~T. Wilkinson.
\newblock Concurrent range reporting in two-dimensional space.
\newblock In {\em Proceedings of the twenty-fifth annual ACM-SIAM Symposium on
  Discrete algorithms}, pages 983--994. SIAM, 2014.

\bibitem{agarwal1990partitioning}
P.~K. Agarwal.
\newblock Partitioning arrangements of lines {II}: Applications.
\newblock {\em Discrete \& Computational Geometry}, 5(6):533--573, 1990.

\bibitem{agarwal1991geometric}
P.~K. Agarwal.
\newblock {\em Geometric partitioning and its applications}.
\newblock Duke University, 1991.

\bibitem{Bentley.79}
J.~L. Bentley.
\newblock Decomposable searching problems.
\newblock {\em Information Processing Letters}, 8(5):244 -- 251, 1979.

\bibitem{chan:revisit}
T.~M. Chan, K.~G. Larsen, and M.~P\u{a}tra\c{s}cu.
\newblock Orthogonal range searching on the {RAM}, revisited.
\newblock In {\em Proceedings of the twenty-seventh annual symposium on
  Computational geometry}, pages 1--10, 2011.

\bibitem{chan2018orthogonal}
T.~M. Chan, Y.~Nekrich, S.~Rahul, and K.~Tsakalidis.
\newblock Orthogonal point location and rectangle stabbing queries in 3-d.
\newblock In {\em 45th International Colloquium on Automata, Languages, and
  Programming, ICALP 2018}, page~31. Schloss Dagstuhl-Leibniz-Zentrum fur
  Informatik GmbH, Dagstuhl Publishing, 2018.

\bibitem{c86}
B.~Chazelle.
\newblock Filtering search: A new approach to query-answering.
\newblock {\em SIAM Journal on Computing}, 15(3):703--724, 1986.

\bibitem{chazelle1993cutting}
B.~Chazelle.
\newblock Cutting hyperplanes for divide-and-conquer.
\newblock {\em Discrete \& Computational Geometry}, 9(2):145--158, 1993.

\bibitem{chazelle1994ray}
B.~Chazelle, H.~Edelsbrunner, M.~Grigni, L.~Guibas, J.~Hershberger, M.~Sharir,
  and J.~Snoeyink.
\newblock Ray shooting in polygons using geodesic triangulations.
\newblock {\em Algorithmica}, 12(1):54--68, 1994.

\bibitem{chazelle1990deterministic}
B.~Chazelle and J.~Friedman.
\newblock A deterministic view of random sampling and its use in geometry.
\newblock {\em Combinatorica}, 10(3):229--249, 1990.

\bibitem{chazelle1986fractionali}
B.~Chazelle and L.~J. Guibas.
\newblock Fractional cascading: {I}. a data structuring technique.
\newblock {\em Algorithmica}, 1(1-4):133--162, 1986.

\bibitem{chazelle1986fractionalii}
B.~Chazelle and L.~J. Guibas.
\newblock Fractional cascading: {II}. applications.
\newblock {\em Algorithmica}, 1(1-4):163--191, 1986.

\bibitem{chazelle2004lower}
B.~Chazelle and D.~Liu.
\newblock Lower bounds for intersection searching and fractional cascading in
  higher dimension.
\newblock {\em Journal of Computer and System Sciences}, 68(2):269--284, 2004.

\bibitem{cormen2009introduction}
T.~H. Cormen, C.~E. Leiserson, R.~L. Rivest, and C.~Stein.
\newblock {\em Introduction to algorithms}.
\newblock MIT press, 2009.

\bibitem{bcko08}
M.~de~Berg, O.~Cheong, M.~van Kreveld, and M.~Overmars.
\newblock {\em Computational Geometry: Algorithms and Applications}.
\newblock Springer-Verlag, 3 edition, 2008.

\bibitem{berg1995cuttings}
M.~de~Berg and O.~Schwarzkopf.
\newblock Cuttings and applications.
\newblock {\em International Journal of Computational Geometry \&
  Applications}, 5(04):343--355, 1995.

\bibitem{de1992two}
M.~de~Berg, M.~van Kreveld, and J.~Snoeyink.
\newblock Two-and three-dimensional point location in rectangular subdivisions.
\newblock In {\em Scandinavian Workshop on Algorithm Theory}, pages 352--363.
  Springer, 1992.

\bibitem{dr91}
P.~F. Dietz and R.~Raman.
\newblock Persistence, amortization and randomization.
\newblock In {\em Proceedings of the second annual ACM-SIAM symposium on
  Discrete algorithms}, page 78–88, 1991.

\bibitem{gk09}
Y.~Giyora and H.~Kaplan.
\newblock Optimal dynamic vertical ray shooting in rectilinear planar
  subdivisions.
\newblock {\em ACM Transactions on Algorithms}, 5(3), 2009.

\bibitem{matouvsek1991cutting}
J.~Matou{\v{s}}ek.
\newblock Cutting hyperplane arrangements.
\newblock {\em Discrete \& Computational Geometry}, 6(3):385--406, 1991.

\bibitem{matousek1995approximations}
J.~Matou{\v{s}}ek.
\newblock Approximations and optimal geometric divide-and-conquer.
\newblock {\em Journal of Computer and System Sciences}, 50(2):203--208, 1995.

\bibitem{mn90}
K.~Mehlhorn and S.~N{\"a}her.
\newblock Dynamic fractional cascading.
\newblock {\em Algorithmica}, 5(2):215--241, 1990.

\bibitem{rahul2014improved}
S.~Rahul.
\newblock Improved bounds for orthogonal point enclosure query and point
  location in orthogonal subdivisions in $\mathbb{R}^3$.
\newblock In {\em Proceedings of the twenty-sixth annual ACM-SIAM Symposium on
  Discrete Algorithms}, pages 200--211. SIAM, 2014.

\bibitem{sleator1983data}
D.~D. Sleator and R.~E. Tarjan.
\newblock A data structure for dynamic trees.
\newblock {\em Journal of Computer and System Sciences}, 26(3):362--391, 1983.

\bibitem{toth2017handbook}
C.~D. Toth, J.~O'Rourke, and J.~E. Goodman.
\newblock {\em Handbook of discrete and computational geometry}.
\newblock CRC press, 2017.

\end{thebibliography}
\bibliographystyle{abbrv}

\begin{appendices}
\section{Proof of \autoref{thm:framework}}
\label{sec:frameworkproof}
\thmframework*
\ignore{
To prove this theorem, we use a variant of the pointer machine model \cite{afshani2012improved}. We give a short description of this model. In this model, the data structure is modelled by an out-degree bounded directed graph. Each vertex in the graph stores an input element and two pointers pointing to other vertices. To answer a query, an algorithm starts at a special vertex, called the root, and explores a subgraph which contains all the elements need to be reported by pointer navigation. Note that random access is disallowed in this model, but other information can be stored or accessed for free. The query time is defined to be the number of vertices of the explored subgraph, and the space is defined to be the number of vertices of the whole directed graph.
}
To prove this theorem, we first show a special property of the subgraph $M_q$ explored to answer a query $q\in{U}$. Note that in the pointer machine model, we begin the exploration with a special cell, called the root. If we consider only the first in-edge to any cell in $M_q$, we obtain a tree.

\begin{lemma}
\label{lem:forkdecomp}
Let $M_q$ be the explored subgraph corresponds to a query $q\in{U}$. We call the memory cells in $M_q$ containing reported ranges marked cells. Let a fork be a subtree of $M_q$ of size at most $c\gamma$ containing two marked cells, where $c$ is a large enough constant and $\gamma$ is the parameter in \autoref{thm:framework}. Then $M_q$ can be decomposed into $\Omega(|R_q|)$ many forks, where $R_q$ is the set of ranges containing $q$.
\end{lemma}

\begin{proof}
The proof we present is very similar to the one described in \cite{afshani2012improved}. 
We generate the forks using the following method. 
For every cell in $M_q$, we assign two values mark and size to it. For marked cell, we initialize its mark value to be one. 
Other cells will have mark value zero. For any cell in $M_q$, we assign one to its size value. 
Without loss of generality, we assume $M_q$ to be a tree. 
At every step, we choose an arbitrary leaf and add its mark value and size value to the corresponding values of its parent. 
Then we remove this cell. If its parent has another child, we repeat this process until its parent becomes a leaf. 
If after this process its parent has mark value two and size value more than $c\gamma$, 
we remove its parent as well and do nothing. We call this situation ``wasted''. 
If its parent has mark value two and size value no more than $c\gamma$, we find a fork. 
We add it to the fork set and remove the subtree. If its parent has mark value less than two, we do nothing. 
Note that its parent cannot have mark value more than two 
because that will indicate one of its child has mark value at least two but not being added to a fork or wasted. 
We go on to the next step until reaching the root. 

Let us consider how many marks will be wasted. 
We only waste marks when we find a subtree containing two marks but of size more than $c\gamma$ 
and when we reach the root with only one mark. 
Since $M_q$ contains $Q(n) +\gamma|R_q|\le2\gamma|R_q|$ cells, 
the number of marks wasted is bounded by $4\gamma|R_q|/(c\gamma)+1=4|R_q|/c+1$. 
Other marks are all stored in forks, 
so the number of forks is more than $(|R_q|-4|R_q|/c-1)/2=\Omega(|R_q|)$ for a sufficiently large $c$.
\end{proof}

We also need another lemma, which follows directed from Lemma 1 in Afshani \cite{afshani2012improved}.

\begin{lemma}
\label{lem:forknum}
The number of forks of size $O(\gamma)$ is $O(S(n)2^{O(\gamma)})$.
\end{lemma}

Now we prove \autoref{thm:framework}.

\begin{proof}
Consider any query point $q\in U$. By definition, it is contained in a set $R_q$ of ranges. 
Consider the explored subgraph $M_q$ when answering $q$. 
By \autoref{lem:forkdecomp}, we can decompose $M_q$ into a set $F_q$ of $\Omega(|R_q|)$ forks 
such that each fork contains two output ranges. 
Note that for the two ranges to be output, $q$ must lie in the intersection of the two ranges. 
Similarly, $q$ must lie in all the intersection of the two ranges for every fork in $M_q$. 
This implies that $q$ is covered by these intersections $\Omega(|R_q|)$ times.

Since we can answer queries for all $q\in U$ and by assumption (i) each $q$ is contained in $t$ ranges,
it implies that the intersections of two ranges in all possible forks cover $U$ $\Omega(t)$ times. 
By \autoref{lem:forknum}, the number of possible forks of size $O(\gamma)$ is $O(S(n)2^{O(\gamma)})$. 
Each fork has ${O(\gamma)\choose{2}}=O(\gamma^2)$ ways to choose two ranges. 
By assumption (ii), the measure of any two ranges is bounded by $v$.
So by a simple measure argument,
$$
O(\gamma^2)S(n)2^{O(\gamma)}v=\Omega(t).
$$
This gives us
$$
S(n)=\Omega(\frac{t}{v2^{O(\gamma)}}).
$$
By our assumption (i), $\gamma t\ge Q(n)$, we also obtain
$$
S(n)=\Omega(\frac{Q(n)}{v2^{O(\gamma)}}).
$$

\end{proof}

\end{appendices}

\end{document}